\documentclass[11pt]{article}
\usepackage[T1]{fontenc}
\usepackage[utf8]{inputenc}
\usepackage[english]{babel}
\usepackage{amsmath}
\usepackage{amssymb,mathrsfs}
\usepackage{enumitem} 
\usepackage{amsthm}
\usepackage{titlesec}
\usepackage[all]{xy}
\usepackage[filecolor=green]{hyperref}
\usepackage{algorithm}
\usepackage{algorithmic}
\usepackage[toc,page]{appendix}
\usepackage{color}
\usepackage{listings}

\usepackage[a4paper]{geometry}
\newtheorem{theorem}{Theorem}[subsection]
\newtheorem{lemma}[theorem]{Lemma}
\newtheorem{proposition}[theorem]{Proposition}

\theoremstyle{definition}
\newtheorem{definition}[theorem]{Definition}
\newtheorem{remark}[theorem]{Remark}
\newtheorem{example}[theorem]{Example}

\def\l@paragraph{\@tocline{3}{0pt}{1pc}{9pc}{}}

\newcommand\K{\mathbb{K}}
\newcommand\N{\mathbb{N}}
\newcommand\M[1]{\mathscr{#1}}

\newcommand\id[1]{\text{Id}_{#1}}
\newcommand{\G}{Gr\"obner}
\newcommand\spol[1]{\text{sp}\left(#1\right)}
\newcommand\X[1]{X^{\left(#1\right)}}
\newcommand\F{\longrightarrow}

\newcommand\lm[1]{\text{lm}\left(#1\right)}
\newcommand\supp[1]{\text{supp}\left(#1\right)}
\newcommand\lgen[1]{\text{lg}\left(#1\right)}
\newcommand\lc[1]{\text{lc}\left(#1\right)}

\newcommand{\RO}{\textbf{RO}\left(G,\ <\right)}
\newcommand\noy[1]{\text{ker}^{-1}\left(#1\right)}
\newcommand\im[1]{\text{im}\left(#1\right)}
\newcommand\red[1]{\text{red}\left(#1\right)}
\newcommand\nf[1]{\text{nf}\left(#1\right)}
\newcommand\obs[1]{\text{obs}\left(#1\right)}

\titleformat{\subsubsection}[runin]
{\normalfont\bfseries}
{\thesubsubsection.}{.5em}{}[.]
\titlespacing{\subparagraph}
{\parindent}{1.5ex plus .1ex minus .2ex}{1.5pt}
\begin{document}

\title{A lattice formulation\\ 
of the noncommutative $F_4$ procedure}
\author{Cyrille Chenavier~\footnote{\noindent Universit\'e Paris-Est Marne-la-Vallée, cyrille.chenavier@u-pem.fr.}}
\date{}
\maketitle

\begin{abstract}

  We introduce a new procedure for constructing noncommutative \G\ bases using a lattice formulation of completion. This leads to a lattice description of the noncommutative $F_4$ procedure. Our procedure is based on the lattice structure of reduction operators which provides a lattice description of the confluence property. We relate reduction operators to noncommutative \G\ bases, we show the Diamond Lemma for reduction operators and we deduce the lattice interpretation of the $F_4$ procedure. Finally, we illustrate our procedure with a complete example.

\end{abstract}


\tableofcontents

\section{Introduction}

The objective of the paper is to introduce a new procedure for constructing noncommutative \G\ bases which turns out to be a lattice formulation of the noncommutative $F_4$ procedure. This formulation is based on a description of the completion procedure using linear algebra techniques and is motivated by the development of effective methods in homological algebra using such techniques~\cite{MR846601, MR1608711, MR3334140, guiraud:hal-01006220, MR2110434, MR0265437}.

The $F_4$ procedure is an improvement of the Buchberger's one where several $S$-polynomials are reduced into normal forms simultaneously. Improvements and optimisations of Buchberger's procedure were first introduced in the context of polynomial ideals, where selections strategies~\cite{MR2781934, MR1784751, MR1263871} and criteria for avoiding useless critical pairs~\cite{MR575678, MR2094215, MR2035234, MR2159476, MR2723052} were investigated. The $F_4$ completion procedure was also introduced for polynomial ideals~\cite{MR1700538}, it is adapted to the noncommutative case~\cite{xiu2012non} and an implementation of this adaptation can be found in the system MAGMA.

Our lattice formulation of $F_4$ uses the approach due to Bergman~\cite{MR506890} who described reduction systems over noncommutative algebras by \emph{reduction operators}. The latter admit a lattice structure inducing lattice formulations of confluence and completion that we present now.

\paragraph{Lattice formulations of confluence and completion.}

A reduction operator relative to a well-ordered set $\left(G,\ <\right)$ is an idempotent linear endomorphism $T$ of the $\K{}$-vector space $\K{G}$ spanned by $G$ such that for every $g\ \notin\ \im{T}$, $T(g)$ is a linear combination of elements of $G$ strictly smaller than $g$. We denote by $\RO$ the set of reduction operators relative to $\left(G,\ <\right)$. 

From~\cite[Proposition 2.1.14]{MR3673007}, the kernel map induces a bijection between $\RO$ and subspaces of $\K{G}$, so that $\RO$ admits a lattice structure defined in terms of kernels:
\begin{itemize}
\item[\textbf{i.}] $T_1\preceq T_2\ $ if $\ \ker\left(T_2\right)\ \subseteq\ \ker\left(T_1\right)$,
\item[\textbf{ii.}] $T_1\wedge T_2\ =\ \ker^{-1}\left(\ker\left(T_1\right)+\ker\left(T_2\right)\right)$,
\item[\textbf{iii.}] $T_1\vee T_2\ =\ \ker^{-1}\left(\ker\left(T_1\right)\cap\ker\left(T_2\right)\right)$.
\end{itemize}
Given a subset $F$ of $\RO$, we denote by $\wedge F$ the lower-bound of $F$, that is the reduction operator whose kernel is the sum of kernels of elements of $F$. We get the following lattice formulation of confluence: $F$ is said to be \emph{confluent} if the image of $\wedge F$ is equal to the intersection of images of elements of $F$. Recall from~\cite[Corollary 2.3.9]{MR3673007} that $F$ is confluent if and only if the reduction relation on $\K{G}$ defined by $v\ \F\ T(v)$ for every $T\in F$ and every $v\notin\im{T}$ is confluent. Moreover, recall from~\cite[Theorem 3.2.6]{MR3673007} that the completion of $F$ is done by the operator $C^F\ =\ \left(\wedge F\right)\vee\left(\vee\overline{F}\right)$, where $\overline{F}$ is a subset of $\RO$ defined from $F$ and $\vee\overline{F}$ is the upper-bound of $\overline{F}$, that is $F\cup\{C^F\}$ is a confluent subset of $\RO$.

In Section~\ref{proceduree de complétion en termes d'opérateurs de réduction}, the operator $C^F$ is used to reduce simultaneously several $S$-polynomials into normal forms using a triangular process such as the $F_4$ procedure does. For that, we introduce \emph{presentations by operators} which relate reduction operators to noncommutative \G\ bases.

\paragraph{Reduction operators and presentations of algebras.}

A presentation by operator of an associative \textbf{A} is a triple $\left(X,\ <,\ S\right)$, where $X$ is a set, $<$ is a monomial order on the set of noncommutative monomials $X^*$ and $S$ is a reduction operator relative to $\left(X^*,\ <\right)$ such that $\textbf{A}$ is isomorphic to the quotient of the free algebra over $X$ by the two-sided ideal spanned by $\ker\left(S\right)$.

In order to describe all the reductions induced by $S$ we consider the "extensions" of $S$, that is the operators which applied to a monomial $w_1w_2w_3$ gives $w_1S(w_2)w_3$. The presentation $\left(X,\ <,\ S\right)$ is said to be \emph{confluent} if the set of extensions of $S$ is a confluent subset of $\RO$. From~\cite[Proposition 3.3.10]{MR3673007}, the presentation $\left(X,\ <,\ S\right)$ is confluent if and only if the set of elements $w\ -\ S(w)$ with $w\ \notin\ \im{S}$ is a noncommutative \G\ basis of $I\left(\ker\left(S\right)\right)$. This link between reduction operators and noncommutative \G\ bases enables us to show the Diamond Lemma in terms of reduction operators in Proposition~\ref{Characterisation of confluent presentations}.

Our procedure for constructing confluent presentations by operators, and thus noncommutative \G\ bases, is given in Section~\ref{Énoncé de l'proceduree de complétion}. At the step number $d$ of the procedure, we reduce the $S$-polynomials of the current presentation $\left(X,\ <,\ S^d\right)$ into normal forms using a set of reduction operators $F_d$. The operator at the step $d+1$ is $S^{d+1}\ =\ S^d\wedge C^{F_d}$. Denoting by $\overline{S}$ the lower-bound of all the operators $S^d$, the triple $\left(X,\ <,\ \overline{S}\right)$ is called the \emph{completed presentation} of \textbf{A}. The main result of the paper is Theorem~\ref{proceduree de Buchberger par opérateur} which asserts that a completed presentation is confluent. In Section~\ref{Exemples annexes}, we show how to implement our procedure with a complete example as an illustration.

\begin{center}
\large\textbf{Organisation of the paper}
\end{center}

Section~\ref{Lattice structure and completion} is a recollection of results from~\cite{MR3673007}: we recall the definitions and properties of reduction operators, their confluence and completion used in the sequel. In Section~\ref{Presentations by operator}, we define presentations by operators, the confluence property of such presentations, we formulate and we show the Diamond Lemma for reduction operators. In Section~\ref{Énoncé de l'proceduree de complétion}, we write our completion procedure and define completed presentations. In Section~\ref{Correction de l'proceduree de complétion}, we show that a completed presentation is confluent. In Section~\ref{Exemples annexes}, we illustrate our completion procedure with a complete example based on the computation of lattice operations of reduction operators.

\section{Reduction operators}\label{Section reduction operators}

\subsection{Lattice structure of reduction operators}\label{Lattice structure and completion}

Throughout the paper, $\K{}$ denotes a commutative field. Given a set $G$, we denote by $\K{G}$ the vector space spanned by $G$. Given a well-order $<$ on $G$, the leading generator of a nonzero element $v\ \in\ \K{G}$ is written $\lgen{v}$. We extend the order $<$ on $G$ into a partial order on $\K{G}$ in the following way: we have $u<v$ if $u\ =\ 0$ and $v\ \neq\ 0$ or if $\text{lg}(u)\ <\ \text{lg}(v)$.

\begin{definition}

A \emph{reduction operator relative to} $\left(G,\ <\right)$ is an idempotent endomorphism $T$ of $\K{G}$ such that for every $g\ \in\ G$, we have $T(g)\ \leq\ g$. We denote by $\textbf{RO}\left(G,\ <\right)$ the set of reduction operators relative to $\left(G,\ <\right)$. Given $T\ \in\ \RO$, a generator $g\ \in\ G$ is said to be a \emph{T-normal form} or \emph{T-reducible} according to $T(g)\ =\ g$ or $T(g)\ \neq\ g$, respectively. We denote by $\nf{T}$ the set of $T$-normal forms and by $\red{T}$ the set of $T$-reducible generators.

\end{definition}

\paragraph{Lattice structure, confluence and completion.}

Recall from~\cite[Proposition 2.1.14]{MR3673007} that the restriction of the kernel map $T\ \longmapsto\ \ker\left(T\right)$ to $\RO$ is a bijection. Using the inverse $\ker^{-1}$, the set $\RO$ admits a lattice structure for the operations 
\begin{itemize}
\item[\textbf{i.}] $T_1\ \preceq\ T_2$ if $\ker\left(T_2\right)\ \subseteq\ \ker\left(T_1\right)$,
\item[\textbf{ii.}] $T_1\wedge T_2\ =\ \noy{\ker\left(T_1\right) + \ker\left(T_2\right)}$,
\item[\textbf{iii.}] $T_1\vee T_2\ =\ \noy{\ker\left(T_1\right)\cap\ker\left(T_2\right)}$.
\end{itemize}
Recall from~\cite[Lemma 2.1.18]{MR3673007} that we have the following implication
\begin{equation}\label{Inclusion of images}
T_1\ \preceq\ T_2\ \Longrightarrow\ \nf{T_1}\ \subseteq\ \nf{T_2}~\footnote{In~\cite{MR3673007}, the notation $\red{T}$ stands for \emph{reduced generators} and correspond to $\nf{T}$ in the present paper. The notation $\red{T}$ of the present paper corresponds to $\text{nred}(T)$ of~\cite{MR3673007} which means \emph{nonreduced generators}.}.
\end{equation}

Given a nonempty subset $F$ of $\RO$, we denote by $\nf{F}$ and $\wedge F$ the set of normal forms for each $T\ \in\ F$ and the lower-bound of $F$, respectively. From (\ref{Inclusion of images}), $\nf{\wedge F}$ is included in $\nf{T}$ for every $T\in F$, so that $\nf{\wedge F}$ is included in $\nf{F}$. We write 
\begin{equation} \label{Definition of obstructions}
\obs{F}\ =\ \nf{F}\setminus\nf{\wedge F}.
\end{equation}
The set $F$ is said to be \emph{confluent} if $\obs{F}$ is the empty set. In Section~\ref{Correction de l'proceduree de complétion}, we use two characterisations of the confluence property in terms of reduction operators. First, recall from~\cite[Theorem 2.2.5]{MR3673007} that $F$ is confluent if and only if it has the \emph{Church-Rosser property}, that is for every $v\in\K{G}$, there exist $T_1,\ \cdots,\ T_r\ \in\ F$ such that $\left(\wedge F\right)(v)\ =\ \left(T_r\circ\cdots\circ T_1\right)(v)$. Moreover, from~\cite[Proposition 2.2.12]{MR3673007}, $F$ is confluent if and only if it is \emph{locally confluent}, that is for every $v\ \in\ \K{G}$ and for every $\left(T,\ T'\right)\ \in\ F\times F$, there exist $v'\in\K{G}$ and $T_1,\ \cdots ,\ T_r,\ T'_1,\ \cdots ,\ T'_k\ \in\ F$ such that $v'\ =\ \left(T_r\circ\cdots\circ T_1\right)\left(T(v)\right)$ and $v'\ =\ \left(T'_k\circ\cdots\circ T'_1\right)\left(T'(v)\right)$. Finally, we recall how a set of reduction operators is completed into a confluent one.

\begin{definition}\label{definition of complement}
  A \emph{complement of F} is an element $C$ of $\RO$ such that
\begin{itemize}
\item[\textbf{i.}]\label{equivalence relation} $\left(\wedge F\right)\wedge C=\wedge F$,
\item[\textbf{ii.}]\label{minimality} $\obs{F}\ \subseteq\ \red{C}$.
\end{itemize}
The \emph{F-complement} is the operator $C^F\ =\ \left(\wedge F\right)\vee\left(\vee\overline{ F}\right)$, where $\vee\overline{F}$ is equal to $\ker^{-1}\left(\K{\nf{F}}\right)$. 
\end{definition}

Recall from~\cite[Proposition 3.2.2]{MR3673007} that a reduction operator $C$ satisfying $\left(\wedge F\right)\wedge C=\wedge F$ is a complement of $F$ if and only if $F\cup\{C\}$ is confluent. Recall from~\cite[Theorem 3.2.6]{MR3673007} that the $F$-complement is a complement of $F$.

\subsection{Presentations by operators}\label{Presentations by operator}

In this section, we relate the confluence property for reduction operators to noncommutative \G\ bases and we prove the Diamond Lemma for reduction operators.

Given a set $X$, we denote by $X^*$ the set of noncommutative monomials over $X$ and we identify the free algebra over $X$ with $\K{X^*}$, equipped with the multiplication induced by concatenation of monomials. A \emph{monomial order} over $X^*$ is a well-founded total strict order $<$ on $X^*$ such that the following conditions are fulfilled:
\begin{itemize}
\item[\textbf{i.}] $1\ <\ w$ for every monomial $w$ different from 1,
\item[\textbf{ii.}] for every $w_1,\ w_2,\ w,\ w'\in X^*$ such that $w<w'$, we have $w_1ww_2\ <\ w_1w'w_2$.
\end{itemize}
For any $f\in\K{X^*}$, the leading monomial of $f$ is written $\lm{f}$ instead of $\lgen{f}$.

\begin{definition}\label{def pres by op}

A \emph{presentation by operator} of an associative algebra \textbf{A} is a triple $\left(X,\ <,\ S\right)$ where
\begin{itemize}
\item[\textbf{i.}] $X$ is a set and $<$ is a monomial order on $X^*$,
\item[\textbf{ii.}]  $S$ is a reduction operator relative to $\left(X^*,\ <\right)$ such that $\textbf{A}$ is isomorphic to $\K{X^*}/I\left(\ker(S)\right)$, where $I\left(\ker(S)\right)$ is the two-sided ideal spanned by $\ker\left(S\right)$.
\end{itemize}

\end{definition}

We fix an algebra \textbf{A} together with a presentation by operator $\left(X,\ <,\ S\right)$ of \textbf{A}. For every integer $n$, we denote by $\X{n}$ and $\X{\leq n}$ the set of monomials of length $n$ and of length smaller or equal to $n$, respectively. For every integers $n$ and $m$ such that $(n,\ m)$ is different from $(0,\ 0)$, we consider the reduction operator
\[S_{n,m}\ =\ \id{\K{\X{\leq n+m-1}}}\ \oplus\ \Big(\id{\K{\X{n}}}\otimes S\otimes\id{\K{\X{m}}}\Big).\]
Explicitly, for every $w\in X^*$, $S_{n,m}(w)$ is defined by: if the length of $w$ is strictly smaller than $n+m$, then $S_{n,m}(w)\ =\ w$, else we let $w=w_1w_2w_3$ where $w_1$ and $w_3$ have length $n$ and $m$, respectively and we have $S_{n,m}(w)=w_1S(w_2)w_3$. We also let $S_{0,0}=S$. 

\begin{definition}

  The set of all the operators $S_{n,m}$ with $(n,\ m)\ \in\ \N^2$, is called the \emph{reduction family} of $\left(X,\ <,\ S\right)$. The presentation $\left(X,\ <,\ S\right)$ is said to be \emph{confluent} if its reduction family is a confluent subset of $\textbf{RO}\left(X^*,\ <\right)$.

\end{definition}

Recall from~\cite[Proposition 3.3.10]{MR3673007} that $\left(X,\ <,\ S\right)$ is confluent if and only if the set of elements $w\ -\ S(w)$ with $w\ \in\ \red{S}$ is a noncommutative \G\ basis of $I\left(\ker(S)\right)$, that is $\red{S}$ spans leading monomials of $I$ as a monomial ideal.

\begin{example}\label{Definition of the braided monoid}

  Let $X=\{x,\ y,\ z\}$ and let $<$ be the deg-lex order induced by $x\ <\ y\ <\ z$. Consider the algebra presented by $\left(X,\ <,\ S\right)$ where $S$ is defined on the basis $X^*$ by $S(yz)\ =\ x$, $S(zx)\ =\ xy$ and $S(w)\ =\ w$ for every monomial $w$ different from $yz$ and $zx$. We have
\[\begin{split}
yxy-xx\ &=\ \left(yxy-yzx\right)\ -\ \left(xx-yzx\right)\\
&=\ \left(yS(zx)-yzx\right)\ -\ \left(S(yz)x-yzx\right)\\
&=\ A\ +\ B\\
\end{split}\]
where $A\ =\ \left(S_{1,0}-\id{\K{X^*}}\right)(yzx)$ and $B\ =\ \left(\id{\K{X^*}}-S_{0,1}\right)(yzx)$. Hence, $yxy\ -\ xx$ belongs to $\ker\left(\wedge F\right)$ where $F$ is the reduction family of the presentation, so that $yxy$ is $\wedge F$-reducible. Moreover, $yxy$ belongs to $\nf{F}$, so that $yxy$ belongs to $\obs{F}$ and $F$ is not confluent. Thus, $\left(X,\ <,\ S\right)$ is not a confluent presentation of \textbf{A}.

\end{example}

In Section~\ref{Énoncé de l'proceduree de complétion} we formulate our procedure for constructing confluent presentations by operators using \emph{critical branchings} that we introduce in Definition~\ref{Critical Branchings}. These branchings are analogous to \emph{ambiguities} for \G\ bases. An ambiguity with respect to $<$ of a subset $R$ of $\K{X^*}$ is a tuple $b=(w_1,\ w_2,\ w_3,\ f,\ g)$ where $w_1,\ w_2,\ w_3$ are monomials such that $w_2\ \neq\ 1$, $f,\ g$ belong to $R$ and one of the following two conditions is fulfilled:
\begin{enumerate}
\item\label{chevauchement} $w_1w_2\ =\ \lm{f}$ and $w_2w_3\ =\ \lm{g}$.
\item\label{inclusion} $w_1w_2w_3\ =\ \lm{f}$ and $w_2\ =\ \lm{g}$.
\end{enumerate}
The $S$-polynomial of $b$ is written $\spol{b}$, that is $\spol{b}\ =\ fw_3\ -\ w_1g$ or $\spol{b}\ =\ f\ -\ w_1gw_3$ according to $b$ is of the form~\ref{chevauchement} or~\ref{inclusion}, respectively. The ambiguity $b$ is said to be \emph{solvable relative to} $<$ if there exists a decomposition
\begin{equation}\label{S w decompo}
\spol{b}\ =\ \sum_{i=1}^n\lambda_iw_if_i{w'}_i,
\end{equation}
where, for every $i\in\{1,\cdots,n\}$, $\lambda_i$ is a non-zero scalar, $w_i,\ {w'}_i$ are monomials and $f_i$ is an element of $R$ such that $w_i\lm{f_i}{w'}_i\ <\ w_1w_2w_3$. The Diamond Lemma~\cite[Theorem 1.2]{MR506890} asserts that $R$ is a noncommutative \G\ basis of $I(R)$ if and only if every critical branching of $R$ with respect to $<$ is solvable relative to $<$.

Our purpose is to formulate and to prove the Diammond Lemma for reduction operators. Until the end of the section, we fix some notations: \textbf{A} is an associative algebra and $\left(X,\ <,\ S\right)$ is a presentation by operator of \textbf{A}. For every pair of integers $(n,\ m)$, we consider the operator $S_{n,m}$ defined such as the beginning of the section. We denote by $R$ the set of elements $w\ -\ S(w)$ with $w\ \in\ \red{S}$. 

\begin{definition}\label{Critical Branchings}
A \emph{critical branching} of $\left(X,\ <,\ S\right)$ is a triple $b\ =\ \left(w,\ (n,\ m),\ (n',\ m')\right)$ where $w$ is a monomial and $(n,\ m)$ and $(n',\ m')$ are couples of integers 
such that
\begin{itemize}
\item[\textbf{i.}] $w$ belongs to $\red{S_{n,m}}\cap\red{S_{n',m'}}$,
\item[\textbf{ii.}] $n\ =\ 0$ or $n'\ =\ 0$,
\item[\textbf{iii.}] $m\ =\ 0$ or $m'\ =\ 0$,
\item[\textbf{iv.}] $n+n'+m+m'$ is strictly smaller than the length of $w$.
\end{itemize}
The \emph{S-polynomial} of $b$ is $\text{SP}(b)\ =\ S_{n,m}(w)\ -\ S_{n',m'}(w)$ and the \emph{source} of $b$ is the monomial $w$.

\end{definition}

\begin{remark}

The roles of $(n,\ m)$ and $(n',\ m')$ being symmetric, we do not distinguish $\left(w,\ (n,\ m),\ (n',\ m')\right)$ and $\left(w,\ (n',\ m'),\ (n,\ m)\right)$.

\end{remark}

\begin{definition}

Let $w\in X^*$ and let $f\in\K{X^*}$. We say that $f$ admits a $\left(S,\ w\right)$-\emph{type decomposition} if it admits a decomposition
\[f\ =\ \sum_{i=1}^n\lambda_i w^1_i\left(w_i-S(w_i)\right)w^2_i,\]
where, for every $i\in\{1,\cdots,n\}$, $\lambda_i$ is a non-zero scalar, $w^1_i$, $w^2_i$ and $w_i$ are monomials such that $w_i$ belongs to $\red{S}$ and $w^1_iw_iw^2_i\ <\ w$.

\end{definition}

\begin{lemma}\label{Lemma before characterisation of confluent presentations}

 There is a one-to-one correspondence $b\ \longmapsto\ \tilde{b}$ between critical branchings of $\left(X,\ <,\ S\right)$ and ambiguities of R with respect to $<$. Moreover, a critical branching b of source w admits a $(S,w)$-type decomposition if and only if $\tilde{b}$ is solvable relative to $<$.

\end{lemma}

\begin{proof}

Let us show the first part of the lemma. Let $b=\left(w,\ (n,\ m),\ (n',\ m')\right)$ be a critical branching of $\left(X,\ <,\ S\right)$. In order to define $\tilde{b}$, we distinguish four cases depending on the values of $n$ and $m$:

\paragraph{Case 1:} $(n,\ m)\ =\ (0,\ 0)$. We write $w\ =\ w_1w_2w_3$, where the lengths of $w_1$ and $w_3$ are equal to $n'$ and $m'$, respectively. By definition of a critical branching, $w$ and $w_2$ belong to $\red{S}$ and we let $\tilde{b}\ =\ \Big(w_1,\ w_2,\ w_3,\ w-S(w),\ w_1\left(w_2-S(w_2)\right)w_3\Big)$. By definition of a critical branching, $n+n'+m+m'\ =\ n'+m'$ is strictly smaller than the length of $w$. In particular, $w_2$ is not the empty word, so that the tuple $\tilde{b}$ is an ambiguity of $R$ with respect to $<$ of the form~\ref{inclusion}.

\paragraph{Case 2:} $n\ =\ 0$ and $m\ \neq\ 0$. By definition of a critical branching, $m'\ =\ 0$. If $n'$ is also equal to $0$, we have $(n',\ m')\ =\ (0,\ 0)$, so that we exchange the roles of $(n,\ m)$ and $(n',\ m')$ and we recover the first case. If $n'\ \neq\ 0$, we write $w\ =\ w_1w_2w_3$, where the lengths of $w_1$ and $w_3$ are equal to $n'$ and $m$, respectively. In particular, $b$ being a critical branching, the monomials $w_1w_2$ and $w_2w_3$ belong to $\red{S}$ and $w_2$ is different from $1$. Hence, $\tilde{b}\ =\ \Big(w_1,\ w_2,\ w_3,\ w_1w_2-S(w_1w_2),\ w_2w_3-S(w_2w_3)\Big)$, is an ambiguity of $R$ with respect to $<$.

\paragraph{Case 3:} $n\ \neq\ 0$ and $m\ =\ 0$. By definition of a critical branching, $n'$ is equal to $0$. Exchanging the roles of $(n,\ m)$ and $(n',\ m')$, we recover the second case.

\paragraph{Case 4:} $n\ \neq\ 0$ and $m\ \neq\ 0$. By definition of a critical branching, the pair $(n',\ m')$ is equal to $(0,\ 0)$. Exchanging the roles of $(n,\ m)$ and $(n',\ m')$, we recover the first case.

We have a well-defined map $b\longmapsto\tilde{b}$ between critical branchings of $\left(X,\ <,\ S\right)$ and ambiguities of $R$ with respect to $<$. Now, we define the inverse map $\tilde{b}\longmapsto b$. Let $\tilde{b}\ =\ \left(w_1,\ w_2,\ w_3,\ f,\ g\right)$ be an ambiguity of $R$ with respect to $<$ and let $w\ =\ w_1w_2w_3$. 
\begin{itemize}
\item If $\tilde{b}$ is an ambiguity of the form~\ref{chevauchement}, let $n$ and $m'$ be the lengths of $w_1$ and $w_3$, respectively. The word $w_2$ being non-empty, $n+m'$ is strictly smaller than the length of $w$, so that $b\ =\ \left(w,\ (n,\ 0),\ (0,\ m')\right)$ is a critical branching of $\left(X,\ <,\ S\right)$.
\item If $\tilde{b}$ is of the form~\ref{inclusion}, let $n$ and $m$ be the lengths of $n$ and $m$, respectively. Then, $b\ =\ \left(w,\ (n,\ m),\ (0,\ 0)\right)$ is a critical branching of $\left(X,\ <,\ S\right)$.
\end{itemize}
Such defined, the two composites of $b\longmapsto\tilde{b}$ and $\tilde{b}\longmapsto b$ are identities.

Let us show the second part of the lemma. Given a critical branching $b$, $\spol{b}$ and $\spol{\tilde{b}}$ are equal. Letting $w$ the source of $w$, a $(S,\ w)$-type decomposition of $\spol{b}$ is precisely a decomposition of the from (\ref{S w decompo}). That shows the second part of the lemma.

\end{proof}

The Diamond Lemma for reduction operators is formulated as follows:

\begin{proposition}\label{Characterisation of confluent presentations}

The presentation $\left(X,\ <,\ S\right)$ is confluent if and only if for every critical branching b of source w, $\emph{SP}(b)$ admits a $\left(S,\ w\right)$-type decomposition.

\end{proposition}

\begin{proof}
The two-sided ideal $I(R)$ spanned by $R$ is equal to $I\left(\ker(S)\right)$. Hence, from~\cite[Proposition 3.3.10]{MR3673007}, $\left(X,\ <,\ S\right)$ is confluent if and only if $R$ is a noncommutative \G\ basis of $I(R)$. From the Diamond Lemma, the presentation $\left(X,\ <,\ S\right)$ is confluent if and only if every ambiguity of $R$ with respect to $<$ is solvable relative to $<$. Thus, from Lemma~\ref{Lemma before characterisation of confluent presentations}, $\left(X,\ <,\ S\right)$ is confluent if and only if for every critical branching $b$ of source $w$ the $S$-polynomial $\spol{b}$ admits a $\left(S,\ w\right)$-type decomposition.
\end{proof}

\begin{example}\label{Spol of braided monoid}

Considering the presentation of Example~\ref{Definition of the braided monoid}, we have one critical branching $b_1\ =\ \left(yzx,\ (1,\ 0),\ (0,\ 1)\right)$ and we have $\spol{b_1}\ =\ yxy\ -\ xx$. This $S$-polynomial does not admit a $(S,\ yzx)$-type decomposition so that we recover that the presentation is not confluent.

\end{example}

\section{Completion procedure}\label{proceduree de complétion en termes d'opérateurs de réduction}

In Section~\ref{Énoncé de l'proceduree de complétion}, we formulate our procedure for constructing confluent presentations by operators and we show the correctness of this procedure in Section~\ref{Correction de l'proceduree de complétion}. Throughout Section~\ref{proceduree de complétion en termes d'opérateurs de réduction}, we fix the following notations:

\begin{itemize}
\item[\textbf{i.}] \textbf{A} is an algebra and $\left(X,\ <,\ S\right)$ is a presentation by operator of \textbf{A}.
\item[\textbf{ii.}] Given a reduction operator $T\ \in\ \textbf{RO}\left(X^*,\ <\right)$ and a pair of integers $(n,\ m)$, the operator $T_{n,m}$ is defined such as the beginning of Section~\ref{Presentations by operator}.
\item[\textbf{iii.}] For every $f\ \in\ \K{X^*}$, we write $T(f)\ =\ \ker^{-1}\left(\K{f}\right)$. Explicitly, $\left(T(f)\right)(\lm{f})$ is equal to $\lm{f}\ -\ 1/\lc{f}f$ and all other monomial is a normal form for $T(f)$. Moreover, we write $\supp{f}$ the \emph{support} of $f$, that is the set monomials occurring in the decomposition of $f$ with a nonzero coefficient.
\item[\textbf{iv.}] Given a subset $E\ \subseteq\ \K{X^*}$, we write $\lm{E}$ the set of leading monomials of elements of $E$.
\end{itemize}

\subsection{Formulation}\label{Énoncé de l'proceduree de complétion}

Our procedure requires a function called {\tt normalisation} with inputs a finite set $E\ \subset\ \K{X^*}$ and a reduction operator $U\ \in\ \textbf{RO}\left(X^*,\ <\right)$ and with output a finite set of reduction operators. Then, {\tt normalisation}$\left(E,\ U\right)$ is defined as follows:
\begin{enumerate}
\item Let $M\ =\ \left(\bigcup_{f\in E}\supp{f}\right)\setminus\lm{E}$ and $F\ =\ \left\{T(f)\ \mid\ f\in E\right\}$.
\item {\tt while} $\exists\ w_1ww_2\ \in\ M$ such that $w\ \in\ \red{U}$, 
\begin{itemize}
\item[\textbf{i.}] we add $T\left(w_1(w-U(w))w_2\right)$ to $F$,
\item[\textbf{ii.}] we remove $w_1ww_2$ from $M$,
\item[\textbf{iii.}] we add $\supp{w_1U(w)w_2}$ to $M$.
\end{itemize}
\item {\tt normalisation}$\left(E,\ U\right)$ is the set $F$ obtained when the loop {\tt while} is over.
\end{enumerate}
The loop {\tt while} is terminating beacause $E$ is finite and $<$ is a monomial order.

We formulate our completion procedure. We assume that the presentation $\left(X,\ <,\ S\right)$ is \emph{finite}, that is $X$ is finite and $\ker(S)$ is finite-dimensional. In particular, the set of critical branchings of $\left(X,\ <,\ S\right)$ is finite.

\vspace{5cm}

\begin{algorithm}
 
\caption{Completion procedure} 
 
\textbf{Initialisation}: 
\begin{itemize}
\item $d\ :=\ 0$,
\item $S^d\ :=\ S$,
\item $Q_d\ :=\ \emptyset$ and $P_d\ :=\ \Big\{\text{critical branchings of}\ \left(X,\ <,\ S^d\right)\Big\}$,
\item $E_d\ :=\ \Big\{w\ -\ {S^{d}}_{n,m}(w)\ \mid\ (w,\ (n,\ m),\ (n',\ m'))\ \in\ P_d\Big\}$.
\end{itemize}

\vspace{3mm}

\begin{algorithmic}[1]

\WHILE{$Q_d\ \neq\ P_d$}

\vspace{2mm}

\STATE $F_d\ :=\ {\tt normalisation}(E_d,\ S^{d})$;

\vspace{2mm}

\STATE $S^{d+1}\ :=\ S^{d}\wedge C^{F_{d}}$;

\vspace{2mm}

\STATE $Q_{d+1}\ :=\ P_{d}$;

\vspace{2mm}

\STATE $d\ =\ d+1$;

\vspace{2mm}

\STATE $P_{d}\ :=\ \Big\{\text{critical branchings of}\ \left(X,\ <,\ S^d\right)\Big\}$;

\vspace{2mm}

\STATE $E_d\ :=\ \Big\{w\ -\ {S^{d}}_{n,m}(w)\ \mid\ (w,\ (n,\ m),\ (n',\ m'))\in P_d\ \setminus\ Q_d\Big\}$;

\vspace{2mm}

\ENDWHILE

\end{algorithmic}
 
\end{algorithm}

This first and the last instruction of the loop \textbf{while} make sense because we have the following:

\begin{lemma}\label{Remarks procedure}
  
Let $d$ be an integer.
\begin{enumerate}
\item\label{Aspect effectif de l'algo de complétion} The kernels of $S^d$ and $C^{F_d}$ are finite-dimensional.
\item\label{Croissance des paires critiques} The set $Q_d$ is included in $P_d$.
\end{enumerate}

\end{lemma}
\begin{proof}

We show Point~\ref{Aspect effectif de l'algo de complétion} by induction on $d$. The kernel of $S^0\ =\ S$ is finite-dimensional by hypotheses. Let $d\ \in\ \N$ and assume that the kernel of $S^d$ is finite-dimensional. Let $M_d\ =\ \bigcup_{f\in E_d}\supp{f}$ be the union of words appearing in $E_d$. The elements of $F_d$ are only acting on $M_d$, so that we have the inclusion
\begin{equation}\label{relation des M_d}
\ker\left(C^{F_d}\right)\ \subset\ \K{M_d}.
\end{equation}
The kernel of $S^d$ being finite-dimensional by induction hypothesis, the set of critical branchings of $\left(X,\ <,\ S^d\right)$ is finite. Hence, $E_d$ and $M_d$ are finite sets, so that $\ker\left(C^{F_d}\right)$ is finite-dimensional from (\ref{relation des M_d}). Moreover, by definition of  $\wedge$, $\ker\left(S^{d+1}\right)$ is equal to $\ker\left(S^d\right)\ +\ \ker\left(C^{F_d}\right)$, so that $\ker\left(S^{d+1}\right)$ is finite-dimensional.

Let us show Point~\ref{Croissance des paires critiques}. By construction, $Q_d$ is equal to $P_{d-1}$, that is $Q_d$ is the set of critical branchings of $\left(X,\ <,\ S^{d-1}\right)$. Let $\left(w,\ (n,\ m),\ (n',\ m')\right)$ be such a critical branching, so that we have
\begin{equation}\label{first relation}
w\ \in\ \red{\left(S^{d-1}\right)_{n,m}}\cap\red{\left(S^{d-1}\right)_{n',m'}}.
\end{equation}
Moreover, by construction, we have $S^{d}\ \preceq\ S^{d-1}$. Hence, from implication (\ref{Inclusion of images}) (see page~\pageref{Inclusion of images}), we have
\begin{equation}\label{second relation}
\red{S^{d-1}}\ \subset\ \red{S^{d}}.
\end{equation}
From (\ref{first relation}) and (\ref{second relation}), $w$ belongs to $\red{{S^{d}}_{n,m}}\cap\red{{S^{d}}_{n',m'}}$, so that $\left(w,\ (n,\ m),\ (n',\ m')\right)$ is a critical branching of $\left(X,\ <,\ S^{d+1}\right)$, that is it belongs to $P_d$. Thus, $Q_d$ is included in $P_d$.

\end{proof}

\begin{remark}\label{effectiveness} 

Our procedure requires to compute lower-bound of reduction operators relative to $\left(X^*,\ <\right)$. In Section~\ref{Exemples annexes}, we give the implementation of $\ker^{-1}$ for totally ordered finite sets, so that it cannot be used for a set of monomials. However, from Lemma~\ref{Remarks procedure}, the kernels of $S^d$ and $C^{F_d}$ are finite-dimensional, so that these two operators can be computed by restrictions over finite-dimensional subspaces of $\K{X^*}$. We illustrate how works such computations in Section~\ref{Exemples annexes}.

\end{remark}

Our procedure has no reason to terminate since there exist finitely presented algebras with no finite \G\ basis~\cite[Section 1.3]{MR1299371}. If the procedure terminates after $d$ iterations of the loop \textbf{while}, we let $S^n\ =\ S^d$ for every integer $n\  \geq\ d$, so that the sequence $\left(S^d\right)_{d\in\N}$ is well-defined if the procedure terminates or not. We let
\[\overline{S}=\bigwedge_{d\in\N}S^d.\]

\begin{definition}

 The triple $\left(X,\ <,\ \overline{S}\right)$ is called the \emph{completed presentation} of $\left(X,\ <,\ S\right)$.
  
\end{definition}

The purpose of the next section is to show that the completed presentation of $\left(X,\ <,\ S\right)$ is a confluent presentation of \textbf{A}, that is our procedure computes a noncommutative \G\ basis.

\subsection{Soundness}\label{Correction de l'proceduree de complétion}

In this section, we say reduction operator instead of reduction operator relative to $\left(X^*,\ <\right)$.

\begin{lemma}\label{Lemme principal sur les S-pol}

Let $w\ \in\ X^*$ and let T and $T'$ be two reduction operators such that $T'\ \preceq\ T$.
\begin{enumerate}
\item\label{n et m fixés, d change} Let $(n,\ m)$ be a pair of integers such that $w$ is $T_{n,m}$-reducible. Then, $\left(T_{n,m}\ -\ T'_{n,m}\right)(w)$ admits a $(T',\ w)$-type decomposition.
\item\label{n et m changent, d fixé} Let $f\ \in\ \K{X^*}$ admitting a $(T,\ w)$-type decomposition. Then, f admits a $(T',\ w)$-type decomposition.
\end{enumerate}

\end{lemma}

\begin{proof}

Let us show Point~\ref{n et m fixés, d change}. We let $w\ =\ w^{(n)}w'w^{(m)}$, where $w^{(n)}$ and $w^{(m)}$ have length $n$ and $m$, respectively. Let
\begin{equation}\label{T(w')}
T(w')\ =\ \sum_{i=1}^k\lambda_iw_i,
\end{equation}
be the decomposition of $T(w')$ with respect to the basis $X^*$. By hypotheses, $T'$ is smaller than $T$, that is $\ker\left(T\right)\ \subseteq\ \ker\left(T'\right)$, so that $T'\circ T$ is equal to $T'$. Hence, we have
\[\begin{split}
\left(T_{n,m}-T'_{n,m}\right)(w)\ &=\ w^{(n)}\left(T(w')-T'(w')\right)w^{(m)}\\
&=\ w^{(n)}\left(T(w')-T'\left(T(w')\right)\right)w^{(m)}.
\end{split}\]
From (\ref{T(w')}), we obtain
\begin{equation}\label{type decompo of the Spol}
\left(T_{n,m}\ -\ T'_{n,m}\right)(w)\ =\ \sum_{i=1}^k\lambda_i w^{(n)}\left(w_i-T'(w_i)\right)w^{(m)}.
\end{equation}
By hypotheses, $w$ is $T_{n,m}$-reducible, so that $w'$ is $T$-reducible and each $w_i$ is strictly smaller than $w'$ for $<$. The strict order $<$ being monomial, each $w^{(n)}w_iw^{(m)}$ is strictly smaller than $w^{(n)}w'w^{(m)}\ =\ w$, so that (\ref{type decompo of the Spol}) is a $(T',\ w)$-type decomposition of  $\left(T_{n,m}\ -\ T'_{n,m}\right)(w)$.

Let us show Point~\ref{n et m changent, d fixé}. Let
\begin{equation}\label{Petite decompo}
f\ =\ \sum_{i=1}^n\lambda_i w^1_i\left(w_i-T(w_i)\right)w^2_i,
\end{equation}
be a $(T,\ w)$-type decomposition of $f$. Letting
\[A\ =\ \sum_{i=1}^n\lambda_i w^1_i\left(w_i-T'(w_i)\right)w^2_i\ \ \text{and}\ \ B\ =\ \sum_{i=1}^n\lambda_i w^1_i\left(T(w_i)-T'(w_i)\right)w^2_i,\]
$f$ is equal to $A\ -\ B$. The decomposition (\ref{Petite decompo}) being $(T,\ w)$-type, each $w'_i\ =\ w^1_iw_iw^2_i$ is strictly smaller than $w$, so that $A$ is $(T',\ w)$-type. For every $i\ \in\ \{1,\ \cdots,\ n\}$, let $n_i$ and $m_i$ be the lengths of $w^1_i$ and $w^2_i$, respectively, so that we have $B\ =\ \sum_{i=1}^n\lambda_i\left(T_{n_i,m_i}-T'_{n_i,m_i}\right)(w'_i)$. Each $w_i$ being $T$-reducible, each $w'_i$ is $T_{n_i,m_i}$-reducible. Hence, from Point~\ref{n et m fixés, d change} of the lemma, each $\left(T_{n_i,m_i}\ -\ T'_{n_i,m_i}\right)(w'_i)$ admits a $(T',\ w'_i)$-type decomposition, so that it admits a $(T',\ w)$-type decomposition since $w'_i$ is strictly smaller than $w$. Hence, $B$ admits a $(T',\ w)$-type decomposition, so that $f$ also admits such a decomposition.

\end{proof}

\paragraph{Notation.}

For every integer $d$, let $F_d$ be the reduction family of $\left(X,\ <,\ S^d\right)$, that is $F_d$ is equal to $ \Big\{\left(S^d\right)_{n,m}\ \mid\ (n,\ m)\ \in\ \N^2\Big\}$.

\begin{lemma}\label{Lemme sur les S-pol}

Let d be an integer, let $(w,\ (n,\ m),\ (n',\ m'))\ \in\ P_d\setminus Q_d$ and let f be the S-polynomial of $(w,\ (n,\ m),(n',\ m'))$.
\begin{enumerate}
\item\label{the S-polynomial belongs to the kernel of the lower bound} $\left(\wedge F_d\right)(f)$ is equal to $0$.
\item\label{the S-polynomial admits an upper decomposition} f admits a $\left(S^{d+1},\ w\right)$-type decomposition.
\end{enumerate}

\end{lemma}

\begin{proof}

  Let us show Point~\ref{the S-polynomial belongs to the kernel of the lower bound}. The two elements $w\ -\ \left(S^{d}\right)_{n,m}(w)$ and $w\ -\ \left(S^{d}\right)_{n',m'}(w)$ belong to $E_d$ by construction of the latter. Hence, by definition of the function {\tt normalisation}, the operators $T_1\ =\ T\left(w-\left(S^{d}\right)_{n,m}(w)\right)$ and $T_2\ =\ T\left(w-\left(S^{d}\right)_{n',m'}(w)\right)$ belong to $F_d$, so that $f\ =\ (w-{S^{d}}_{n,m}(w))-(w-{S^{d}}_{n',m'}(w))$ belongs to the kernel of $T_1\wedge T_2$. The latter is included in the kernel of $\wedge F_d$, which shows Point~\ref{the S-polynomial belongs to the kernel of the lower bound}.

Let us show Point~\ref{the S-polynomial admits an upper decomposition}. The operator $C^{F_{d}}$ being a complement of $F_{d}$, we have
\begin{equation}\label{Formule du complément dans Buchberger}
\wedge\left(F_{d}\cup\left\{C^{F_{d}}\right\}\right)\ =\ \wedge F_{d},
\end{equation}
and $F_{d}\cup\left\{C^{F_{d}}\right\}$ is confluent (see the paragraph after Definition~\ref{definition of complement}), that is it has the Church-Rosser property (see the paragraph before Definition~\ref{definition of complement}). Hence, from Point~\ref{the S-polynomial belongs to the kernel of the lower bound} of the lemma and Relation (\ref{Formule du complément dans Buchberger}), there exist $T_1,\ \cdots,\ T_r\ \in\ F_{d}\cup\left\{C^{F_{d}}\right\}$ such that
\begin{equation}\label{Réécriture de f et 0}
\left(T_r\circ\cdots\circ T_1\right)\left(f\right)\ =\ 0.
\end{equation}
We let $f_1\ =\ \left(\id{\K{X^*}}\ -\ T_1\right)(f)$ and for every $k\ \in\ \{2,\ \cdots,\ r\}$, $f_k\ =\ \left(\id{\K{X^*}}\ -\ T_k\right)\left(T_{k-1}\circ\cdots\circ T_1(f)\right)$. From (\ref{Réécriture de f et 0}), we have
\begin{equation}\label{Décompo de f en fi}
f\ =\ \sum_{k=1}^{r}f_k.
\end{equation}
The tuple $(w,\ (n,\ m),\ (n',\ m'))$ being a critical branching of $\left(X,\ <,\ S^d\right)$, $w$ belongs to $\red{\left(S^d\right)_{n,m}}\cap\red{\left(S^d\right)_{n',m'}}$, so that the leading monomial of $f$ is strictly smaller than $w$. Moreover, each $T_i$ is either of the form $T\left(w_1(w_2\ -\ S^d(w_2))w_3\right)$, or is equal to $C^{F_d}$. Hence, each $f_i$ admits a  $\left(S^d,\ w\right)$-type decomposition or a $\left(C^{F_d},\ w\right)$-type decomposition. The reduction operators $S^d$ and $C^{F_d}$ being smaller than $S^{d+1}$, each $f_i$ admits a $\left(S^{d+1},\ w\right)$-type decomposition from Point~\ref{n et m changent, d fixé} of Lemma~\ref{Lemme principal sur les S-pol}, so that $f$ admits a $\left(S^{d+1},\ w\right)$-type decomposition from (\ref{Décompo de f en fi}).

\end{proof}

\begin{proposition}\label{type composition pour Sd}

Let d be an integer. For every $\left(w,\ (n,\ m),\ (n',\ m')\right)\ \in\ Q_d$, the $S$-polynomial $\left(S^d\right)_{n,m}(w)\ -\ \left(S^d\right)_{n',m'}(w)$ admits a $\left(S^d,\ w\right)$-type decomposition.

\end{proposition}

\begin{proof}

We show the proposition by induction on $d$. The set $Q_0$ being empty, Proposition~\ref{type composition pour Sd} holds for $d\ =\ 0$. Assume that for every $\left(w,\ (n,\ m),\ (n',\ m')\right)\ \in\ Q_d$, ${S^d}_{n,m}(w)\ -\ {S^d}_{n',m'}(w)$ admits a $\left(S_d,\ w\right)$-type decomposition. We let
\[\begin{split}
&A\ =\ \left(S^d\right)_{n',m'}(w)\ -\ \left(S^{d+1}\right)_{n',m'}(w)\ \ ,\\
&B\ =\ \left(S^d\right)_{n,m}(w)\ -\ \left(S^{d+1}\right)_{n,m}(w),\\
&C\ =\ \left(S^d\right)_{n,m}(w)\ -\ \left(S^{d}\right)_{n',m'}(w).
\end{split}\]
We have
\[\left(S^{d+1}\right)_{n,m}(w)\ -\ \left(S^{d+1}\right)_{n',m'}(w)\ =\ A-B+C.\]
By construction, $S^{d+1}$ is smaller than $S^d$. Moreover, $\left(w,\ (n,\ m),\ (n',\ m')\right)$ being a critical branching, $w$ belongs to $\red{\left(S^d\right)_{n,m}}\cap\red{\left(S^d\right)_{n',m'}}$. Hence, from Point~\ref{n et m fixés, d change} of Lemma~\ref{Lemme principal sur les S-pol}, $A$ and $B$ admit a $\left(S^{d+1},\ w\right)$-type decomposition. It remains to show that $C$ admits a $\left(S^{d+1},\ w\right)$-type decomposition. By construction, $Q_{d+1}$ is equal to $P_d$, so that it contains $Q_{d}$ from Point~\ref{Croissance des paires critiques} of Lemma~\ref{Remarks procedure}. If $\left(w,\ (n,\ m),\ (n',\ m')\right)$ does not belong to $Q_d$, $C$ admits a $\left(S^{d+1},\ w\right)$-type decomposition from Point~\ref{the S-polynomial admits an upper decomposition} of Lemma~\ref{Lemme sur les S-pol}. If $\left(w\ ,\ (n,\ m),\ (n',\ m')\right)$ belongs to $Q_d$, $C$ admits a $\left(S^d,\ w\right)$-type decomposition by induction hypothesis. Hence, from Point~\ref{n et m changent, d fixé} of Lemma~\ref{Lemme principal sur les S-pol}, $C$ admits a $\left(S^{d+1},\ w\right)$-type decomposition. 

\end{proof}

Recall that the lower-bound of the operators $S^d$ is written $\overline{S}$. The last lemma we need to prove Theorem~\ref{proceduree de Buchberger par opérateur} is

\begin{lemma}\label{Dernier lemme}

\begin{enumerate}
\item\label{Idéaux engendrés} The sequence $\left(I_d\right)_{d\in\N}$ of ideals spanned by $\ker\left(S^d\right)$ is constant.
\item\label{Non réduit de la complétion} $\emph{Red}\left(\overline{S}\right)$ is equal to $\bigcup_{d\in\N}\emph{Red}\left(S^d\right)$.
\end{enumerate}

\end{lemma}

\begin{proof}

Let us show Point~\ref{Idéaux engendrés}. By definition of the function {\tt normalisation}, the kernel of each element of $F_d$ is included in $I_d$. In particular, $\ker\left(\wedge F_d\right)\ =\ \sum_{T\in F_d}\ker\left(T\right)$ is also included in $I_d$. Moreover, $C^{F_d}$ being a complement of $F_d$, it is smaller than $\wedge F_d$, that is its kernel is included in the one of $\wedge F_d$. In particular, $\ker\left(C^{F_d}\right)$ is included in $I_d$, so that $\ker\left(S^{d+1}\right)$, which by definition is equal to $\ker\left(S^d\right)\ +\ \ker\left(C^{F_d}\right)$, is also included in $I_d$. Hence, the sequence $\left(I_d\right)_{d\in\N}$ is not increasing. Moreover, the sequence $\left(S^d\right)_{d\in\N}$ is not increasing by construction, which means that $\left(\ker\left(S^d\right)\right)_{d\in\N}$ is not decreasing. Hence, $\left(I_d\right)_{d\in\N}$ constant.

Let us show Point~\ref{Non réduit de la complétion}. The equality we want to prove means that the set $F\ =\ \left\{S^d\ \mid\ d\in\N\right\}$ is confluent. From Newman's Lemma (see the paragraph before Definition~\ref{definition of complement}) in terms of reduction operators, it is sufficient to show that $F$ is locally confluent. Let $f\ \in\ \K{X^*}$ and let $d$ and $d'$ be two integers which we assume to satisfy $d\ \geq\ d'$. In particular, we have $S^{d'}\ \preceq\ S^d$, so that $S^{d}\circ S^{d'}$ is equal to $S^{d'}$. Hence, $\left(S^{d}\circ S^{d'}\right)(f)$ and $S^{d}(f)$ are equal, so that $F$ is locally confluent.

\end{proof}

\begin{theorem}\label{proceduree de Buchberger par opérateur}

Let \emph{\textbf{A}} be an algebra and let $\left(X,\ <,\ S\right)$ be a presentation by operator of \emph{\textbf{A}}. The completed presentation of $\left(X,\ <,\ S\right)$ is a confluent presentation of \emph{\textbf{A}}.

\end{theorem}

\begin{proof}

Let $\overline{S}$ be the lower-bound of the operators $S^d$.

First, we show that $\left(X,\  <,\ \overline{S}\right)$ is a presentation of \textbf{A}. From Point~\ref{Idéaux engendrés} of Lemma~\ref{Dernier lemme}, the ideal spanned by the kernels of the operators $S^d$ is equal to the ideal $I$ spanned by the kernel of $S^0\ =\ S$. In particular, the ideal spanned by $\ker\left(\overline{S}\right)\ =\ \sum_{d\in\N}\ker\left(S^d\right)$ is equal to $I$. Hence, $\left(X,\ <,\ S\right)$ being a presentation of \textbf{A}, $\left(X,\ <,\ \overline{S}\right)$ is also a presentation of \textbf{A}.

Let us show that this presentation is confluent. From the Diamond Lemma, it is sufficient to show that for each critical branching
$b\ =\ \left(w,\ (n,\ m),\ (n',\ m')\right)$ of $\left(X,\ <,\ \overline{S}\right)$, the $S$-polynomial $\spol{b}$ admits a $\left(\overline{S},\ w\right)$-type decomposition. From Point~\ref{Non réduit de la complétion} of Lemma~\ref{Dernier lemme}, there exist integers $d$ and $d'$ such that $w\ \in\ \red{\left(S^d\right)_{n,m}}\cap\red{\left(S^{d'}_{n',m'}\right)}$. Without lost of generalities, we may assume that $d$ is greater or equal to $d'$, so that $b$ is a critical branching of $\left(X,\ <,\ S^d\right)$, that is it belongs to $P_d\ =\ Q_{d+1}$. We let
\[\begin{split}
A_d\ &=\ \left(S^{d+1}\right)_{n',m'}(w)\ -\ \overline{S}_{n',m'}(w),\\
B_d\ &=\ \left(S^{d+1}\right)_{n,m}(w)\ -\ \overline{S}_{n,m}(w),\\
C_d\ &=\ \left(S^{d+1}\right)_{n,m}(w)\ -\ \left(S^{d+1}\right)_{n',m'}(w).
\end{split}\]
We have
\begin{equation}\label{Ultime décomposition du S-pol}
\spol{b}\ =\ A_d\ -\ B_d\ +\ C_d.
\end{equation}
From Proposition~\ref{type composition pour Sd}, $b$ being an element of $Q_{d+1}$, $C_d$ admits a $\left(S^{d+1},\ w\right)$-type decomposition, so that it admits a $\left(\overline{S},\ w\right)$-type decomposition from Point~\ref{n et m changent, d fixé} of Lemma~\ref{Lemme principal sur les S-pol}. Moreover, $S^{d+1}$ being smaller than $S^d$, $w$ belongs to $\red{\left(S^{d+1}\right)_{n,m}}\cap\red{\left(S^{d+1}\right)_{n',m'}}$. The operator $\overline{S}$ being smaller than $S^{d+1}$, $A_d$ and $B_d$ also admit a $\left(\overline{S},\ w\right)$-type decomposition from Point~\ref{n et m fixés, d change} of Lemma~\ref{Lemme principal sur les S-pol}. Hence, from (\ref{Ultime décomposition du S-pol}), $\spol{b}$ admits a $\left(\overline{S},\ w\right)$-type decomposition.

\end{proof}

\begin{example}\label{Example of procedure}

 In Section~\ref{Exemples annexes}, we compute the completed presentation of Example~\ref{Definition of the braided monoid}. It is given by the operator defined by $\overline{S}(yz)\ =\ x$, $\overline{S}(zx)\ =\ xy$, $\overline{S}(yxy)\ =\ xx$, $\overline{S}(yxx)\ =\ xxz$, $\overline{S}(yxxx)\ =\ xxxy$ and $\overline{S}(w)\ =\ w$ for all other monomial $w$.
  
  \end{example}

\subsection{Example}\label{Exemples annexes}

In this section, we compute the completed presentation of Example~\ref{Example of procedure}. Before that, we show how to use Gaussian elimination to compute lattice operations and completion for reduction operators relative to totally ordered finite sets. We use the SageMath software, written in Python.

\paragraph{Lattice operations and completion.}

Let $\left(G,\ <\right)$ be a totally ordered finite set. The set $G$ being finite, the Gaussian elimination provides a unique basis $\M{B}$ of any suspace $V\ \subseteq\ \K{G}$ such that for every $e\ \in\ \M{B}$, $\lc{e}$ is equal to 1 and, given two different elements $e$ and $e'$ of $\M{B}$, $\lgen{e'}$ does not belong to the decomposition of $e$. The operator $T\ =\ \noy{V}$ satisfies $T(\lgen{e})\ =\ \lgen{e}\ -\ e$ for every $e\ \in\ \M{B}$ and $T(g)\ =\ g$ if $g$ is not a leading generator of $\M{B}$. Moreover, we represent the subspaces of $\K{G}$ by lists of generating vectors and for any list of vectors {\tt L}, let {\tt reducedBasis(L)} be the basis of $\K{{\tt L}}$ obtained by Gaussian elimination.

First, we define the function {\tt operator} which takes as input a list of vectors {\tt L} and returns $\noy{\K{{\tt L}}}$. We deduce the functions which compute the lattice operations of $\RO$.

\vspace{0.1cm}

\begin{lstlisting}
def operator(G):
    L=reducedBasis(G)
    n=len(L[0])
    V=VectorSpace(QQ,n)
    v=V.zero()
    G=(lg(L[0])-1)*[v]+[L[0]]
    k=len(L)
    for i in [1..k-1]:
        G=G+(lg(L[i])-lg(L[i-1])-1)*[v]+[L[i]]
    G=G+(n-lg(L[k-1]))*[v]
    return identity_matrix(QQ,n)-matrix(G).transpose()

def lowerBound(T_1,T_2):
    V_1,V_2=kernel(T_1.transpose()),kernel(T_2.transpose())
    G_1,G_2=basis(V_1),basis(V_2)
    L_1,L_2=reducedBasis(G_1),reducedBasis(G_2)
    G=L_1+L_2
    L=reducedBasis(G)
    return operator(L)

def upperBound(T_1,T_2):
    V_1,V_2=kernel(T_1.transpose()),kernel(T_2.transpose())
    V=V_1.intersection(V_2)
    G=basis(V)
    L=reducedBasis(G)
    return operator(L)
\end{lstlisting}

\vspace{0.1cm}

By definition of the $F$-complement, we need an intermediate function with input a reduction operator $T$ and output $\noy{\K{\nf{T}}}$. We define this function before defining the one of the $F$-complement.

\vspace{0.1cm}

\begin{lstlisting}
def tilde(T):
    n,L=T.nrows(),[]
    for i in [0..n-1]:
        j,k=i,n-i-1
        if T[i,i]==1: L=L+[vector(j*[0]+[1]+k*[0])]
    return operator(L)
 
def complement(L):
    n,C,T=len(L),L[0],tilde(L[0])
    for i in [1..n-1]: C=lowerBound(C,L[i])
    for j in [1..n-1]: T=upperBound(T,tilde(L[j]))
    return lowerBound(C,T)
\end{lstlisting}

\paragraph{Example.}

Now, we use our implementation to compute the completed presentation of Example~\ref{Definition of the braided monoid}: we consider the algebra \textbf{A} presented by $\left(X,\ <,\ S\right)$ where $X=\{x,\ y,\ z\}$, $<$ is the deg-lex order induced by $x\ <\ y\ <\ z$ and $S(yz)\ =\ x$, $S(zx)\ =\ xy$ and $S(w)\ =\ w$ for every monomial $w$ different from $yz$ and $zx$.

Recall that $S^d$ denotes the operator of the presentation at the beginning of step $d$ of the procedure, $P_d$ is the set of critical branchings of $\left(X,\ <,\ S^d\right)$, $Q_d\ =\ P_{d-1}$, $E_d\ =\ \Big\{w\ -\ {S^{d}}_{n,m}(w)\ \mid\ (w,\ (n,\ m),\ (n',\ m'))\ \in\ P_d\setminus Q_d\Big\}$ and $F_d\ =\ {\tt normalisation}(E_d,\ S^d)$. Moreover, we represent reduction operators by matrices. For that, we use that the operators appearing in the procedure act nontrivially on finite-dimensional subspaces of $\K{X^*}$ spanned by an ordered set of monomials $w_1\ <\ w_2\ <\ \cdots\ <\ w_n$.

At the first step, we have $d\ =\ 0$. The presentation $\left(X,\ <,\ S^0\right)$ has one critical branching $b_1\ =\ \left(yzx,\ (1,\ 0),\ (0,\ 1)\right)$ and we have $P_0\ =\ \left\{b_1\right\}$ and $E_0\ =\ \Big\{yzx-xx,\ yzx-yxy\Big\}$. We have $F_0\ =\ \Big\{T_1,\ T_2\Big\}$ where the matrices of the restrictions of $T_1$ and $T_2$ to the subspace spanned by $xx\ <\ yxy\ <\ yzx$ are
\[T_1\ =\ \begin{pmatrix}
1&0&1\\
0&1&0\\
0&0&0
\end{pmatrix}\ \ \text{and}\ \
T_2\ =\ \begin{pmatrix}
1&0&0\\
0&1&1\\
0&0&0
\end{pmatrix}. \]
that is $T_1(yzx)\ =\ xx$ and $T_2(yzx)\ =\ yxy$. The matrice of $C^{F_0}\ =\ \text{complement}\left(\left[T_1,\ T_2\right]\right)$ restricted to $\K{\left\{xx,\ yxy,\ yzx\right\}}$ is
\[\ \begin{pmatrix}
1&1&0\\
0&0&0\\
0&0&1
\end{pmatrix},
\]
The operator $S^1\ =\ S\wedge C^{F_0}$ can be computed by restriction to the subspace spanned by $x\ <\ xx\ <\ xy\ <\ yz\ <\ zx\ <\ yxy$ and the matrices of the restrictions of $S^0$ and $C^{F_0}$ to this subspace are
\[S^0\ =\ \begin{pmatrix}
1&0&0&1&0&0\\
0&1&0&0&0&0\\
0&0&1&0&1&0\\
0&0&0&0&0&0\\
0&0&0&0&0&0\\
0&0&0&0&0&1
\end{pmatrix}\ \ \text{and}\ \
C^{F_0}\ =\ \begin{pmatrix}
1&0&0&0&0&0\\
0&1&0&0&0&1\\
0&0&1&0&0&0\\
0&0&0&1&0&0\\
0&0&0&0&1&0\\
0&0&0&0&0&0
\end{pmatrix}.\]
We obtain that $S^1$ is the operator defined by $S^1(yz)\ =\ x$, $S^1(zx)\ =\ xy$, $S^1(yxy)\ =\ xx$ and $S^1(w)\ =\ w$ for every monomial $w$ different from $yz$, $zx$ and $yxy$.

The presentation $\left(X,\ <,\ S^1\right)$ has two new critical branchings $b_2$ and $b_3$ equal to $\left(yxyz,\ (2,\ 0),\ (0,\ 1)\right)$ and $\left(yxyxy,\ (2,\ 0),\ (0,\ 2)\right)$, respectively. We have $P_1\ =\ \left\{b_1,\ b_2,\ b_3\right\}$, $P_1\setminus Q_1\ =\ \left\{b_2,\ b_3\right\}$ and $E_1\ =\ \Big\{yxyz\ -\ xxz,\ yxyz\ -\ yxx,\ yxyxy\ -\ xxxy,\ yxyxy\ -\ yxxx\Big\}$. Moreover, $F_1\ =\ {\tt normalisation}\left(E_1,\ S^1\right)$ is equal to
\[\left\{
\begin{split}
&T_3\ =\ T\left(yxyz-xxz\right),\ T_4\ =\ T\left(yxyz-yxx\right)\\
&T_5\ =\ T\left(yxyxy-xxxy\right),\ T_6\ =\ T\left(yxyxy-yxxx\right)
\end{split}
\right\},\]
where $T(f)\ =\ \noy{\K{f}}$. The restriction of $C^{F_1}$ to $\K{\left\{xxz,\ yxx,\ xxxy,\ yxxx,\ yxyz,\ yxyxy\right\}}$ is
\[C^{F_1}\ =\ \begin{pmatrix}
1&1&0&0&0&0\\
0&0&0&0&0&0\\
0&0&1&1&0&0\\
0&0&0&0&0&0\\
0&0&0&0&1&0\\
0&0&0&0&0&1
\end{pmatrix},\]
and we obtain that $S^2\ =\ S^1\wedge C^{F_1}$ is defined by $S^2(yz)\ =\ x$, $S^2(zx)\ =\ xy$, $S^2(yxy)\ =\ xx$, $S^2(yxx)\ =\ xxz$, $S^2(yxxx)\ =\ xxxy$ and all other monomial is a normal form for $S^2$.

The computation of the operator $C^{F_2}$ gives the identity operator of size 11, which corresponds to the monomials $x^4\ <\ x^3y\ <\ x^2zx\ <\ yx^3\ <\ x^5\ <\ x^3yz\ <\ yx^3z\ <\ yxyx^2\ <\ x^3yxyx\ <\ yx^4y\ <\ yxyx^3$. Hence, no new critical branching is created at this step and the procedure stops.

\newpage

\bibliography{Biblio}

\def\cprime{$'$}
\begin{thebibliography}{10}

\bibitem{MR846601}
David~J. Anick.
\newblock On the homology of associative algebras.
\newblock {\em Trans. Amer. Math. Soc.}, 296(2):641--659, 1986.

\bibitem{MR1608711}
Roland Berger.
\newblock Confluence and {K}oszulity.
\newblock {\em J. Algebra}, 201(1):243--283, 1998.

\bibitem{MR506890}
George~M. Bergman.
\newblock The diamond lemma for ring theory.
\newblock {\em Adv. in Math.}, 29(2):178--218, 1978.

\bibitem{MR2781934}
Anna~M. Bigatti, Massimo Caboara, and Lorenzo Robbiano.
\newblock Computing inhomogeneous {G}r\"obner bases.
\newblock {\em J. Symbolic Comput.}, 46(5):498--510, 2011.

\bibitem{MR1784751}
Miguel~A. Borges-Trenard, Mijail Borges-Quintana, and Teo Mora.
\newblock Computing {G}r\"obner bases by {FGLM} techniques in a non-commutative
  setting.
\newblock {\em J. Symbolic Comput.}, 30(4):429--449, 2000.

\bibitem{MR575678}
Bruno Buchberger.
\newblock A criterion for detecting unnecessary reductions in the construction
  of {G}r\"obner-bases.
\newblock In {\em Symbolic and algebraic computation ({EUROSAM} '79,
  {I}nternat. {S}ympos., {M}arseille, 1979)}, volume~72 of {\em Lecture Notes
  in Comput. Sci.}, pages 3--21. Springer, Berlin-New York, 1979.

\bibitem{MR2094215}
Massimo Caboara, Martin Kreuzer, and Lorenzo Robbiano.
\newblock Efficiently computing minimal sets of critical pairs.
\newblock {\em J. Symbolic Comput.}, 38(4):1169--1190, 2004.

\bibitem{MR3673007}
Cyrille Chenavier.
\newblock Reduction operators and completion of rewriting systems.
\newblock {\em J. Symbolic Comput.}, 84:57--83, 2018.

\bibitem{MR3334140}
Sergio Chouhy and Andrea Solotar.
\newblock Projective resolutions of associative algebras and ambiguities.
\newblock {\em J. Algebra}, 432:22--61, 2015.

\bibitem{MR1700538}
Jean-Charles Faug{\`e}re.
\newblock A new efficient algorithm for computing {G}r\"obner bases {$(F_4)$}.
\newblock {\em J. Pure Appl. Algebra}, 139(1-3):61--88, 1999.
\newblock Effective methods in algebraic geometry (Saint-Malo, 1998).

\bibitem{MR2035234}
Jean-Charles Faug{\`e}re.
\newblock A new efficient algorithm for computing {G}r\"obner bases without
  reduction to zero {$(F_5)$}.
\newblock In {\em Proceedings of the 2002 {I}nternational {S}ymposium on
  {S}ymbolic and {A}lgebraic {C}omputation}, pages 75--83 (electronic). ACM,
  New York, 2002.

\bibitem{MR1263871}
Jean-Charles Faug\`ere, Patrizia Gianni, Daniel Lazard, and Teo Mora.
\newblock Efficient computation of zero-dimensional {G}r\"obner bases by change
  of ordering.
\newblock {\em J. Symbolic Comput.}, 16(4):329--344, 1993.

\bibitem{guiraud:hal-01006220}
Yves Guiraud, Eric Hoffbeck, and Philippe Malbos.
\newblock {Convergent presentations and polygraphic resolutions of associative
  algebras}.
\newblock 65 pages, December 2017.

\bibitem{MR2110434}
Yuji Kobayashi.
\newblock Gr\"obner bases of associative algebras and the {H}ochschild
  cohomology.
\newblock {\em Trans. Amer. Math. Soc.}, 357(3):1095--1124 (electronic), 2005.

\bibitem{MR2159476}
Martin Kreuzer and Lorenzo Robbiano.
\newblock {\em Computational commutative algebra. 2}.
\newblock Springer-Verlag, Berlin, 2005.

\bibitem{MR2723052}
Martin Kreuzer and Lorenzo Robbiano.
\newblock {\em Computational commutative algebra 1}.
\newblock Springer-Verlag, Berlin, 2008.
\newblock Corrected reprint of the 2000 original.

\bibitem{MR1299371}
Teo Mora.
\newblock An introduction to commutative and noncommutative {G}r\"obner bases.
\newblock {\em Theoret. Comput. Sci.}, 134(1):131--173, 1994.
\newblock Second International Colloquium on Words, Languages and Combinatorics
  (Kyoto, 1992).

\bibitem{MR0265437}
Stewart~B. Priddy.
\newblock Koszul resolutions.
\newblock {\em Trans. Amer. Math. Soc.}, 152:39--60, 1970.

\bibitem{xiu2012non}
Xingqiang Xiu.
\newblock {\em Non-commutative Gr{\"o}bner bases and applications}.
\newblock PhD thesis, Universit{\"a}t Passau, 2012.

\end{thebibliography}

\end{document}